\documentclass[12pt]{article}

\usepackage{amsmath,amssymb,amsthm} 
\newtheorem{theorem}{Theorem}
\newtheorem{corollary}{Corollary} 

\begin{document} 
\title{EPR states and Bell correlated states \\
in algebraic quantum field theory} 
\author{Yuichiro Kitajima} 
\date{}
\maketitle

\begin{abstract}
A mathematical rigorous definition of EPR states has been introduced by Arens and Varadarajan for finite dimensional systems, and extended by Werner to general systems. In the present paper we follow a definition of EPR states due to Werner. Then we show that an EPR state for incommensurable pairs is Bell correlated, and that the set of EPR states for incommensurable pairs is norm dense between two strictly space-like separated regions. 
\end{abstract}



\section{Introduction}
Einstein, Podolsky and Rosen (EPR) \cite{EPR} discussed a system consisting of
two particles. They have interacted initially and then moved out so that the positions,
and the momenta, of the two particles are strictly correlated,
respectively.
It follows that if one were to measure the position of the first particle, one could predict with certainty the outcome of a position measurement on the second particle; and it is also the case for a momentum measurement.
EPR proposed a criterion of reality: ``If, without in any way disturbing a system, we can predict with certainty (i.e., with probability equal to unity) the value of a physical quantity, then there exists an element of physical reality corresponding to this physical quantity''
\cite[p.777]{EPR}.

In accordance with this criterion, the position and the momentum of the second particle have simultaneous reality since the measurement on the first particle has not disturbed the second particle. On the other hand, position and momentum cannot have simultaneous reality in any states in quantum mechanics. Therefore EPR regarded quantum-mechanical description as incomplete, and concluded that the quantum mechanical description of a physical system should be supplemented by postulating the existence of ``hidden variables,'' the specification of which would predetermine the result of measuring any observable of the system. Later Bohm simplified EPR original state. It is a unit vector in the Hilbert space $\mathbb{C}^2 \otimes \mathbb{C}^2$, represented as
\[ \frac{1}{\sqrt{2}} \left( 
\begin{pmatrix} 1 \\ 0 \end{pmatrix} \otimes \begin{pmatrix} 0 \\ 1 \end{pmatrix} - 
\begin{pmatrix} 0 \\ 1 \end{pmatrix} \otimes \begin{pmatrix} 1 \\ 0 \end{pmatrix} \right). \]

Bell assumed a complete description in terms of hidden variables and the locality assumptions tacitly assumed in EPR argument, and derived Bell's inequality under these assumptions. Bohm's simplified EPR state, however, violates Bell's inequality.
In the mid 1980s, Summers and Werner obtained a series of deep mathematical results on Bell's inequality in algebraic quantum field theory. For example, in most standard quantum field models, all normal states maximally violate Bell's inequalities across spacelike separated tangent wedges (cf. \cite{Redei1998} and \cite{Summers1990}).

The relation between EPR states and Bell's inequalities has been also investigated (cf. \cite{Halvorson2000}, \cite{Hofer-SzaboVecsernyes2012a}, \cite{Hofer-SzaboVecsernyes2012b}, \cite{Huang2008}). For example, Halvorson \cite{Halvorson2000} defined EPR states as a state of canonical commutation relations algebra, and showed that it maximally violates Bell's inequality. Another formulation of EPR states was introduced by Arens and Varadarajan \cite{ArensVaradarajan2000}, and extended by Werner \cite{Werner1999} to general systems.

In the present paper we follow a definition of EPR states due to Werner, and try to clarify the relations between EPR states and Bell correlated states in algebraic quantum theory. We show that an EPR state for incommensurable pairs is Bell correlated (Theorem \ref{violates}), and that the set of EPR states for incommensurable pairs is norm dense between two strictly space-like separated regions in algebraic quantum field theory (Theorem \ref{exists}).

\section{The relations between EPR states and Bell correlated states}
In this paper, we use the following notation. 
If $\mathbb{A}$ is a set of operators acting on a Hilbert space $\mathcal{H}$, let $\mathbb{A}'$ represent its commutant, the set of all bounded operators on $\mathcal{H}$ which commute with all elements of $\mathbb{A}$. $[X, Y]$ denotes the commutator of $X$ and $Y$, i.e., $[X,Y]=XY-YX$. 
We call $X$ a self-adjoint contraction on a Hilbert space $\mathcal{H}$ if $X$ is a self-adjoint operator on $\mathcal{H}$ such that $-I \leq X \leq I$, where $I$ is an identity operator on $\mathcal{H}$. $\mathfrak{N}_1$ and $\mathfrak{N}_2$ denote von Neumann algebras on a Hilbert space $\mathcal{H}$ such that $\mathfrak{N}_1 \subseteq \mathfrak{N}_2'$, and $\mathfrak{N}_1 \vee \mathfrak{N}_2$ denotes a von Neumann algebra generated by $\mathfrak{N}_1$ and $\mathfrak{N}_2$. 

In this section, we provide an abstract definition of EPR states. Let $\mathfrak{N}$ be a von Neumann algebra and let $A_1$ and $A_2$ be commuting self-adjoint operators in $\mathfrak{N}$.
A normal state $\omega$ of $\mathfrak{N}$ is called {\em an EPR state}\
\cite{Werner1999}
for $(A_1,A_2)$ if 
\[ \omega((A_1-A_2)^2)=0. \]

For any commuting self-adjoint operators $A_1$ and $A_2$, the joint probability distribution $\mu^{A_1,A_2}_{\omega}$ of $A_1,A_2$ in $\varphi$ is defined uniquely to be a probability measure $\mu$ on $\mathbb{R}^2$ such that
$$
\omega(f(A_1,A_2))=\int_{\mathbb{R}^{2}}f(x,y)\,d\mu^{A_1,A_2}_{\omega}(x,y)
$$ 
for any polynomial $f(x,y)$.
Then, $\omega$ is an EPR state if and only if $\mu^{A_1,A_2}_{\omega}$ is concentrated in the diagonal, i.e, $\mu^{A_1,A_2}_{\omega}(\{ (x,x) | x \in \mathbb{R} \})=1$. Thus, simultaneous measurements of $A_1$ and $A_2$ always give concordant results, and each one of the outcomes would predict with certainty the other (cf. \cite{Ozawa2006}). In other words, $A_1$ and $A_2$ are strictly correlated in $\omega$ if $\omega$ is an EPR state for $A_1$ and $A_2$.


Nocommutativity of operators as well as strict correlation plays an important role in EPR's argument. EPR asserted that ``either (1) the quantum-mechanical description of reality given by the wave function is not complete or (2) when the operators corresponding to two physical quantities do not commute the two quantities cannot have simultaneous reality" \cite[p.778]{EPR}. 
In this paper, we interpret the second part of this assertion in terms of a beable algebra. Let $\mathfrak{B}$ be a C*-algebra and let $\omega$ be a state of $\mathfrak{B}$. $\mathfrak{B}$ is called a beable algebra for a given state $\omega$ if there is a probability measure $\mu$ on the space $\mathcal{S}_{DF}(\mathfrak{B})$ of dispersion-free states of $\mathfrak{B}$ satisfying
\[ \omega(A)=\int_{\mathcal{S}_{DF}(\mathfrak{B})} \omega(A)d\mu(A) \]
for every $A \in \mathfrak{B}$ \cite[p.2447]{HalvorsonClifton1999}. A dispersion-free state $\omega$ of $\mathfrak{B}$ satisfies the condition that $\omega(X^2)=\omega(X)^2$ for any self-adjoint element $X \in \mathfrak{B}$. Roughly speaking, a beable algebra for $\omega$ is the set of observable which can be taken to have determinate values statistically distributed in accordance with $\omega$. $\mathfrak{B}$ is a beable algebra  for $\omega$ if and only if $\omega(|[A,B]|^2)=0$ for any $A,B \in \mathfrak{B}$ \cite[Proposition 2.2]{HalvorsonClifton1999}. Thus $\omega(|[A,B]|^2) \neq 0$ means that $A$ and $B$ cannot have simultaneous reality in $\omega$. 

Let consider two operators $A$ and $B$ and a unit vector $\Psi_0$ such that
\[ A=\begin{pmatrix} 1 & 0 & 0 \\ 0 & 0 & 0 \\ 0 & 0 & 1 \end{pmatrix}, B= \frac{1}{2}\begin{pmatrix} 1 & 1 & 0 \\ 1 & 1 & 0 \\ 0 & 0 & 2 \end{pmatrix}, \Psi_0=\begin{pmatrix}  0 \\ 0 \\ 1 \end{pmatrix} . \]
Then $[A,B] \neq 0$. On the other hand, from a beable algebraic point of view, $A$ and $B$ can have simultaneous reality in $\Psi_0$ because $A \Psi_0 = B \Psi_0 = \Psi_0$. Therefore $[A,B] \neq 0$ does not always mean that $A$ and $B$ cannot have simultaneous reality.

On the basis of a bealble algebra, we say a normal state $\omega$ of $\mathfrak{N}_1 \vee \mathfrak{N}_2$ is {\em an EPR state for incommensurable pairs} if there exist projections $E_1, F_1 \in \mathfrak{N}_1$ and $E_2, F_2 \in \mathfrak{N}_2$ such that $\omega$ is an EPR state for $(E_1, E_2)$ and $(F_1, F_2)$, $\omega(|[E_1,F_1]|^2) \neq 0$ and $\omega(|[E_2,F_2]|^2) \neq 0$ (cf \cite{OzawaKitajima2012}). In this state, if we were to measure $E_2$, we could predict with certainty the outcome of $E_1$; and if we were to measure $F_2$, we could predict with certainty the outcome of $F_1$. In accordance with EPR's criterion of reality, $E_1$ and $F_1$ have simultaneous reality. On the other hand $E_1$ and $F_1$ cannot have simultaneous reality since $E_1$ does not commute with $F_1$ in $\omega$. Therefore an EPR state for incommensurable pairs allows for similar arguments to EPR's \cite{EPR}.

If there exist self-adjoint contractions $A_1,B_1 \in \mathfrak{N}_1$ and $A_2,B_2 \in \mathfrak{N}_2$ such that 
\begin{equation}
\label{bell_i}
 \frac{1}{2} \vert \omega(A_1A_2+A_1B_2+B_1A_2-B_1B_2) \vert >1, 
\end{equation}
we say that {\em $\omega$ violates a Bell inequality, or is Bell correlated}.
It is known that the left-hand side of (\ref{bell_i}) cannot exceed $\sqrt{2}$ (cf. \cite{Cirel'son1980} and \cite{SummersWerner1987a}). We say that {\em $\omega$ is strongly maximally correlated} if there exist self-adjoint contractions $A_1,B_1 \in \mathfrak{N}_1$ and $A_2,B_2 \in \mathfrak{N}_2$ such that 
\[ \frac{1}{2} \vert \omega(A_1A_2+A_1B_2+B_1A_2-A_2B_2) \vert =\sqrt{2}. \]

Landau \cite{Landau1987} provided a sufficient condition for the strongly maximal Bell correlation in some normal state, and Bohata and Hamhalter \cite{BohataHamhalter2009} \cite{BohataHamhalter2010} characterized strongly maximally Bell correlated states.



As it is shown in Theorem \ref{maximal} and Corollary \ref{maximal_cor}, a strongly maximally Bell correlated normal state is an EPR state for incommensurable pairs. The following theorem can be proven in a similar way of the proof of \cite[Theorem 2.1 (2b)]{SummersWerner1987a}.

\begin{theorem}
\label{maximal}
Let $\mathfrak{N}_1$ and $\mathfrak{N}_2$ be von Neumann algebras on a Hilbert space $\mathcal{H}$ such that $\mathfrak{N}_1 \subseteq \mathfrak{N}_2'$, let $\omega$ be a normal state of $\mathfrak{N}_1 \vee \mathfrak{N}_2$ and let $S_1$ and $S_2$ be the support of $\omega |_{\mathfrak{N}_1}$ and $\omega |_{\mathfrak{N}_2}$ respectively, where $\omega |_{\mathfrak{N}_i}$ is the restriction of $\omega$ to $\mathfrak{N}_i$. Then the following conditions are equivalent.
\begin{enumerate}
\item $\omega$ is strongly maximally Bell correlated.
\item There are projections $E_1,F_1 \in S_1\mathfrak{N}_1S_1$ and $E_2,F_2 \in S_2\mathfrak{N}_2S_2$ such that
\[ \omega((E_1-E_2)^2)=\omega((F_1-F_2)^2)=0, \]
\[ (2E_1-S_1)(2F_1-S_1)+(2F_1-S_1)(2E_1-S_1)=0. \]
\end{enumerate}
\end{theorem}

\begin{proof}
Let $(\pi, \mathcal{H}_{\omega}, \Omega)$ be GNS representation of $\mathfrak{N}_1 \vee \mathfrak{N}_2$ induced by $\omega$. Then $\pi(S_1)\Omega=\pi(S_2)\Omega=\Omega$.
\begin{description}
\item[$1 \Rightarrow 2$]

Since $\omega$ is strongly maximally correlated, there exists self-adjoint contractions $A_1,B_1 \in \mathfrak{N}_1$ and $A_2, B_2 \in \mathfrak{N}_2$ such that $\frac{1}{2} \vert \omega(A_1A_2+A_1B_2+B_1A_1-B_1B_2) \vert =\sqrt{2}$.
Let $X_1:=(1/2)(S_1A_1S_1+iS_1B_1S_1)$ and $X_2 := (1/2\sqrt{2})(S_2A_2S_2+S_2B_2S_2+i(S_2A_2S_2-S_2B_2S_2))$. Then 
\begin{equation}
\label{e1.01}
 X_1^*X_1+X_1X_1^*=\frac{1}{2}(S_1A_1S_1)^2+\frac{1}{2}(S_1B_1S_1)^2 \leq S_1
\end{equation}
and
\begin{equation}
\label{e1.02}
X_2^*X_2+X_2X_2^* =\frac{1}{2}(S_2A_2S_2)^2+\frac{1}{2}(S_2B_2S_2)^2 \leq S_2. 
\end{equation}
Since $\pi(S_1)\Omega=\pi(S_2)\Omega=\Omega$,
\begin{equation}
\label{e1.1}
 \begin{split}
&\sqrt{2} \times \frac{1}{2} \vert \omega(A_1A_2+A_1B_2+B_1A_2-B_1B_2) \vert \\
&=\frac{\sqrt{2}}{2} \vert \langle \Omega, \pi((S_1A_1S_1)(S_2A_2S_2)+(S_1A_1S_1)(S_2B_2S_2)  \\
& \ \ \ \ \ \ \ \ \ \ \ \ +(S_1B_1S_1)(S_2A_2S_2)-(S_1B_1S_1)(S_2B_2S_2))\Omega \rangle \vert \\
&=4 \vert \text{Re} \langle \Omega, \pi(X_1^*X_2) \Omega \rangle \vert \\
&\leq 2 \vert \text{Re} \langle \pi(X_1) \Omega, \pi(X_2) \Omega \rangle \vert + 2 \vert \text{Re} \langle \pi(X_2^*) \Omega, \pi(X_1^*) \Omega \rangle \vert \\
&=\vert \| \pi(X_1) \Omega \|^2 + \| \pi(X_2) \|^2 - \| \pi(X_1-X_2)\Omega \|^2 \vert \\
& \ \ \ \ \ \ \ \ \ \ \ + \vert \| \pi(X_2^*)\Omega \|^2 + \| \pi(X_1^*) \Omega \|^2 - \| \pi(X_2^* - X_1^*) \|^2 \vert \\
&\leq \langle \Omega, \pi(X_1^*X_1+X_1X_1^*+X_2^*X_2+X_2X_2^*) \Omega \rangle \\
&\leq 2.
\end{split} 
\end{equation}
Since $(1/2)\omega(A_1A_2+A_1B_2+B_1A_2-B_1B_2)=\sqrt{2}$, equality obtains in (\ref{e1.1}). By Equations (\ref{e1.01}) and (\ref{e1.02}),
\begin{equation}
\label{e1.2}
\langle \Omega, \pi(X_1^*X_1+X_1X_1^*) \Omega \rangle = 1,
\end{equation}
\begin{equation}
\label{e1.3}
\langle \Omega, \pi(X_2^*X_2+X_2X_2^*) \Omega \rangle = 1,
\end{equation}
\begin{equation}
\label{e1.4}
\| \pi(X_1 - X_2)\Omega \|^2 = \| \pi(X_1^* - X_2^*) \Omega \|^2 = 0.
\end{equation}
Equations (\ref{e1.01}) and (\ref{e1.2}) imply $\langle \Omega,(S_1A_1S_1)^2 \Omega \rangle=\langle \Omega, (S_1B_1S_1)^2 \Omega \rangle=\langle \Omega, S_1 \Omega \rangle$. Thus
\begin{equation}
\label{e1.5}
(S_1A_1S_1)^2=(S_1B_1S_1)^2=S_1.
\end{equation}
Similarly Equations (\ref{e1.02}) and (\ref{e1.3}) imply
\begin{equation}
\label{e1.6}
(S_2A_2S_2)^2=(S_2B_2S_2)^2=S_2.
\end{equation}
Equation (\ref{e1.4}) entails
\begin{equation}
\label{e1.11}
\pi(S_1A_1S_1)\Omega = \pi \left( \frac{1}{\sqrt{2}}(S_2A_2S_2+S_2B_2S_2) \right) \Omega,
\end{equation}
\begin{equation}
\label{e1.12}
\pi(S_1B_1S_1)\Omega = \pi \left( \frac{1}{\sqrt{2}}(S_2A_2S_2-S_2B_2S_2) \right) \Omega.
\end{equation}
By Equations (\ref{e1.4}), (\ref{e1.5}) and (\ref{e1.6})
\begin{equation}
\label{e1.13}
\begin{split}
\pi(S_1A_1S_1S_1B_1S_1+S_1B_1S_1S_1A_1S_1) \Omega 
&=(2/i) \pi(X_1^2 - X_1^{*2}) \Omega \\
&=(2/i) \pi (X_2^2 - X_2^{*2}) \Omega \\
&=0,
\end{split}
\end{equation}
which implies
\begin{equation}
\label{e1.8}
S_1A_1S_1S_1B_1S_1+S_1B_1S_1S_1A_1S_1=0.
\end{equation}
Let 
\[ E_1 := \frac{1}{2} \left( \frac{1}{\sqrt{2}}(S_1A_1S_1+S_1B_1S_1)+S_1 \right), \ \ \ E_2:=\frac{S_2A_2S_2+S_2}{2}, \]
\[ F_1:=\frac{1}{2} \left( \frac{1}{\sqrt{2}}(S_1A_1S_1-S_1B_2S_1)+S_1 \right), \ \ \ F_2:=\frac{S_2B_2S_2+S_2}{2}. \]
By Equations (\ref{e1.5}), (\ref{e1.6}) and (\ref{e1.8}), $E_1,F_1 \in S_1\mathfrak{N}_1S_1$ and $E_2,F_2 \in S_2\mathfrak{N}_2S_2$ are projections, and $(2E_1-S_1)(2F_1-S_1)+(2F_1-S_1)(2E_1-S_1)=0$. Equations (\ref{e1.11}) and (\ref{e1.12}) imply $\pi(E_1)\Omega=\pi(E_2)\Omega$ and $\pi(F_1)\Omega=\pi(F_2)\Omega$. Therefore $\omega((E_1-E_2)^2)=\omega((F_1-F_2)^2)=0$.

\item[$2 \Rightarrow 1$]
Let define $A_1 := 1/\sqrt{2}(2E_1-S_1)+1/\sqrt{2}(2F_1-S_1)$, $B_1:=1/\sqrt{2}(2E_1-S_1)-1/\sqrt{2}(2F_1-S_1)$, $A_2 := 2E_2-S_2$ and $B_2 := 2F_2 - S_2$.

Since $(2E_1-S_1)(2F_1-S_1)+(2F_1-S_1)(2E_1-S_1)=0$, $A_1^2=B_1^2=S_1$ and $A_2^2=B_2^2=S_2$. Thus $A_1$, $B_1$, $A_2$ and $B_2$ are self-adjoint contractions. $\pi(E_1)\Omega=\pi(E_2)\Omega$ and $\pi(F_1)\Omega=\pi(F_2)\Omega$ imply
\[ \frac{1}{2}\omega(A_1A_2+A_1B_2+B_1A_2-B_1B_2)=\sqrt{2}. \]
\end{description}

\end{proof}

\begin{corollary}
\label{maximal_cor}
Let $\mathfrak{N}_1$ and $\mathfrak{N}_2$ be von Neumann algebras on a Hilbert space $\mathcal{H}$ such that $\mathfrak{N}_1 \subseteq \mathfrak{N}_2'$, let $\omega$ be a normal state of $\mathfrak{N}_1 \vee \mathfrak{N}_2$. If $\omega$ is strongly maximally Bell correlated, then $\omega$ is an EPR state for incommensurable pairs.
\end{corollary}

\begin{proof}
Let $(\pi, \mathcal{H}_{\omega}, \Omega)$ be GNS representation of $\mathfrak{N}_1 \vee \mathfrak{N}_2$ induced by $\omega$, and let $S_1$ and $S_2$ be the support of $\omega |_{\mathfrak{N}_1}$ and $\omega |_{\mathfrak{N}_2}$ respectively. 
By Theorem \ref{maximal}, there are projections $E_1,F_1 \in S_1\mathfrak{N}_1S_1$ and $E_2,F_2 \in S_2\mathfrak{N}_2S_2$ such that $\omega((E_1-E_2)^2)=\omega((F_1-F_2)^2)=0$ and $(2E_1-S_1)(2F_1-S_1)+(2F_1-S_1)(2E_1-S_1)=0$.

Let $A_1:=2E_1-S_1$ and $B_1:=2F_1-S_1$. Then $\omega([A_1,B_1]^*[A_1,B_1])=4$ and $[E_1,F_1]^*[E_1,F_1]=1/16[A_1,B_1]^*[A_1,B_1]$ since $A_1,B_1 \in S_1\mathfrak{N}_1S_1$, $A_1B_1+B_1A_1=0$ and $A_1^2=B_1^2=S_1$. It follows that $\omega([E_1,F_1]^*[E_1,F_1])=1/4$. 

$\pi(E_1)\Omega=\pi(E_2)\Omega$ and $\pi(F_1)\Omega=\pi(F_2)\Omega$ imply $\pi([E_1,F_1])\Omega=-\pi([E_2,F_2])\Omega$. Thus $\omega([E_2,F_2]^*[E_2,F_2])=\omega([E_1,F_1]^*[E_1,F_1])=1/4$. Therefore $\omega$ is an EPR state for incommensurable pairs.

\end{proof}

The converse of Corollary \ref{maximal_cor} does not hold. Let consider a vector state induced by $\Psi_1$.
\[ \Psi_1:= \frac{1}{3} \begin{pmatrix} 1 \\ 0 \\ 0 \end{pmatrix} \otimes \begin{pmatrix} 1 \\ 0 \\ 0 \end{pmatrix} + \frac{1}{3}\begin{pmatrix} 0 \\ 1 \\ 0 \end{pmatrix} \otimes \begin{pmatrix} 0 \\ 1 \\ 0 \end{pmatrix} + \frac{\sqrt{7}}{3} \begin{pmatrix} 0 \\ 0 \\ 1 \end{pmatrix} \otimes \begin{pmatrix} 0 \\ 0 \\ 1 \end{pmatrix} \]

Then this state is an EPR state for 
\[ \left( \begin{pmatrix} 1 & 0 & 0 \\ 0 & 0 & 0 \\ 0 & 0 & 0 \end{pmatrix} \otimes \begin{pmatrix} 1 & 0 & 0 \\ 0 & 1 & 0 \\ 0 & 0 & 1 \end{pmatrix}, \begin{pmatrix} 1 & 0 & 0 \\ 0 & 1 & 0 \\ 0 & 0 & 1 \end{pmatrix} \otimes \begin{pmatrix} 1 & 0 & 0 \\ 0 & 0 & 0 \\ 0 & 0 & 0 \end{pmatrix} \right) \]
and
\[ \left( \frac{1}{2} \begin{pmatrix} 1 & 1 & 0 \\ 1 & 1 & 0 \\ 0 & 0 & 0 \end{pmatrix} \otimes  \begin{pmatrix} 1 & 0 & 0 \\ 0 & 1 & 0 \\ 0 & 0 & 1 \end{pmatrix}, \frac{1}{2}\begin{pmatrix} 1 & 0 & 0 \\ 0 & 1 & 0 \\ 0 & 0 & 1 \end{pmatrix} \otimes \begin{pmatrix} 1 & 1 & 0 \\ 1 & 1 & 0 \\ 0 & 0 & 0 \end{pmatrix} \right), \]
so it is an EPR state for incommensurable pairs.

Let define
\[ \begin{split}
\mathbb{A}_1 := \{ A \otimes I \in \mathbb{B}(\mathcal{H}_3) \otimes I | &\text{there exists} \ I \otimes B \in I \otimes \mathbb{B}(\mathcal{H}_3) \ \text{such that} \\
& \ (A \otimes I)\Psi_1=(I \otimes B)\Psi_1 \},
\end{split} \]
where $\mathbb{B}(\mathcal{H}_3)$ is the set of all operators on the 3-dimensional Hilbert space.
Then any projection in $\mathbb{A}_1$ can be expressed as
\[ \begin{pmatrix} a_{11} & a_{12} & 0 \\ a_{21} & a_{22} & 0 \\ 0 & 0 & a_{33} \end{pmatrix} \otimes I, \]
where $a_{33}=0,1$. Thus for any projections $E,F \in \mathbb{A}_1$, $(2E-I)(2F-I)+(2F-I)(2E-I) \neq 0$. Therefore the vector state induced by $\Psi_1$ is not strongly maximally Bell correlated state by Theorem \ref{maximal} although it is an EPR state for incommensurable pairs.

In the following theorem, we examine whether an EPR state for incommensurable pairs is Bell correlated or not.

\begin{theorem}
\label{violates}
Let $\mathfrak{N}_1$ and $\mathfrak{N}_2$ be von Neumann algebras on a Hilbert space $\mathcal{H}$ such that $\mathfrak{N}_1 \subseteq \mathfrak{N}_2'$ and let $\omega$ be a normal state of $\mathfrak{N}_1 \vee \mathfrak{N}_2$. If $\omega$ is an EPR state for incommensurable pairs, then $\omega$ is Bell correlated.
\end{theorem}

\begin{proof}
Let $\omega$ is an EPR state for incommensurable pair. Then there exist projections $E_1, F_1 \in \mathfrak{N}_1$ and $E_2, F_2 \in \mathfrak{N}_2$ such that $\omega$ is an EPR state for incommensurable pairs $(E_1, E_2)$ and $(F_1, F_2)$.

Let $(\pi,\mathcal{H}_{\omega}, \Omega)$ be a GNS representation of $\mathfrak{N}_1 \vee \mathfrak{N}_2$ induced by $\omega$ and let 
\[ c:=\langle \pi(E_1^{\perp}F_1E_1) \Omega, \pi(E_1^{\perp}F_1E_1) \Omega \rangle + \langle \pi(E_1F_1E_1^{\perp}) \Omega, \pi(E_1F_1E_1^{\perp}) \Omega \rangle. \]
Since $\omega([E_1,F_1]^*[E_1,F_1]) \neq 0$, $\pi([E_1,F_1]) \Omega \neq 0$. It implies that $\pi(E_1^{\perp}F_1E_1) \Omega \neq 0$ or $\pi(E_1F_1E_1^{\perp}) \Omega \neq 0$. Thus $c \neq 0$.

Let 
\[ A_1:=E_1-E_1^{\perp}, \]
\[ B_1:=E_1F_1E_1^{\perp}+E_1^{\perp}F_1E_1, \]
\[ A_2:=\frac{1}{\sqrt{1+c^2}}(E_2-E_2^{\perp})+\frac{c}{\sqrt{1+c^2}}(E_2F_2E_2^{\perp}+E_2^{\perp}F_2E_2), \]
\[ B_2:=\frac{1}{\sqrt{1+c^2}}(E_2-E_2^{\perp})-\frac{c}{\sqrt{1+c^2}}(E_2F_2E_2^{\perp}+E_2^{\perp}F_2E_2). \]
Then $A_1$, $B_1$, $A_2$ and $B_2$ are self-adjoint contractions because $A_1^2,B_1^2,A_2^2,B_2^2 \leq I$, where $I$ is an identity operator on $\mathcal{H}$.

By $\pi(E_1)\Omega =\pi(E_2)\Omega$ and $\pi(F_1)\Omega=\pi(F_2)\Omega$,
\[ \begin{split}
\omega((E_1-E_1^{\perp})(E_2-E_2^{\perp})) &
=\langle \Omega, \pi((E_1-E_1^{\perp})(E_2-E_2^{\perp})) \Omega \rangle \\
&=\langle \Omega, \pi((E_1-E_1^{\perp})(E_1-E_1^{\perp})) \Omega \rangle \\
&=1,
\end{split} \]

\[ \begin{split}
&\omega((E_1F_1E_1^{\perp}+E_1^{\perp}F_1E_1)(E_2F_2E_2^{\perp}+E_2^{\perp}F_2E_2)) \\
&=\langle \Omega, \pi((E_1F_1E_1^{\perp}+E_1^{\perp}F_1E_1)(E_2F_2E_2^{\perp}+E_2^{\perp}F_2E_2)) \Omega \rangle \\
&=\langle \Omega, \pi((E_1F_1E_1^{\perp}+E_1^{\perp}F_1E_1)(E_1^{\perp}F_1E_1+E_1F_1E_1^{\perp})) \Omega \rangle \\
&=\langle \Omega, \pi((E_1F_1E_1^{\perp}F_1E_1+E_1^{\perp}F_1E_1F_1E_1^{\perp})) \Omega \rangle \\
&=c. \end{split} \]

Therefore
\[ \frac{1}{2}\omega(A_1A_2+A_1B_2+B_1A_2-B_1B_2)=\sqrt{1+c^2}>1. \]

\end{proof}

The converse of Theorem \ref{violates} does not hold.
Let consider a vector state induced by $\Psi_2$.
\[ \Psi_2 := \frac{1}{2} \begin{pmatrix} 1 \\ 0 \end{pmatrix} \otimes \begin{pmatrix} 1 \\ 0 \end{pmatrix} + \frac{\sqrt{3}}{2} \begin{pmatrix} 0 \\ 1 \end{pmatrix} \otimes \begin{pmatrix} 0 \\ 1 \end{pmatrix} \]
A vector state induced by $\Psi_2$ is a Bell correlated state \cite{Gisin1991}.

Let define
\[ \begin{split}
\mathbb{A}_2 := \{ A \otimes I \in \mathbb{B}(\mathcal{H}_2) \otimes I | &\text{there exists} \ I \otimes B \in I \otimes \mathbb{B}(\mathcal{H}_2) \ \text{such that} \\
& \ (A \otimes I)\Psi_2=(I \otimes B)\Psi_2 \},
\end{split} \]
where $\mathbb{B}(\mathcal{H}_2)$ is the set of all operators on the 2-dimensional Hilbert space. Then any projection in $\mathbb{A}_2$ can be expressed as
\[ \begin{pmatrix} a_{11} &  0 \\ 0 & a_{22}  \end{pmatrix} \otimes I. \]
Thus for any operators $E \otimes I,F \otimes I \in \mathbb{A}_2$, $[E \otimes I, F \otimes I ]\Psi_2=0$. Therefore the vector state induced by $\Psi_2$ is not an EPR state for incommensurable pairs although it is a Bell correlated state.

\section{Existence of EPR states for incommensurable pairs in algebraic quantum field theory}
In this section, we shall consider algebraic quantum field theory.
In algebraic quantum field theory, each bounded open region $\mathcal{O}$ in the Minkowski space is associated with a von Neumann algebra $\mathfrak{N}(\mathcal{O})$. Such a von Neumann algebra is called a local algebra.
We say that bounded open regions $\mathcal{O}_1$ and $\mathcal{O}_2$ are
strictly space-like separated
if there is a neighborhood $\mathcal{V}$ of the origin of the Minkowski space
such that $\mathcal{O}_1 + \mathcal{V}$ and $\mathcal{O}_2$ are space-like separated.

In the present paper, we make the following assumptions.
For any bounded open region $\mathcal{O}$ in the Minkowski space, $\mathfrak{N}(\mathcal{O})$ is properly infinite \cite[Corollary 1.11.6]{Baumgartel1995}.
If $\mathcal{O}_1$ and $\mathcal{O}_2$ are space-like separated,
then $[X_1,X_2]=0$
for any $X_1 \in \mathfrak{N}(\mathcal{O}_1)$ and $X_2 \in \mathfrak{N}(\mathcal{O}_2)$. If $\mathcal{O}_1$ and $\mathcal{O}_2$ are strictly space-like separated, then $X_1X_2 \neq 0$ for any nonzero operators $X_1 \in \mathfrak{N}(\mathcal{O}_1)$ and $X_2 \in \mathfrak{N}(\mathcal{O}_2)$ \cite[Theorem 1.12.3]{Baumgartel1995}.

The following theorem shows that there is a dense set of EPR states for incommensurable pairs between two strictly space-like separated regions.


\begin{theorem}
\label{exists}
Let $\mathfrak{N}_1$ and $\mathfrak{N}_2$ be properly infinite
von Neumann algebras on a Hilbert space $\mathcal{H}$ such that $\mathfrak{N}_1 \subseteq \mathfrak{N}_2'$ and $X_1X_2 \neq 0$ for any nonzero operators $X_1 \in \mathfrak{N}_1$ and $X_2 \in \mathfrak{N}_2$.

For any unit vector $\Psi$ and any real number $\epsilon > 0$, there exists a unit vector $\Psi'$ such that $\| \Psi - \Psi' \| < \epsilon$ and a vector state induced by $\Psi'$ is an EPR state for incommensurable pairs.
\end{theorem}

\begin{proof}
Since $\mathfrak{N}_1$ and $\mathfrak{N}_2$ are properly infinite, there are a countably infinite families $\{ E_{1,i} \in \mathfrak{N}_1 | i \in \mathbb{N} \}$ and $\{ E_{2,i} \in \mathfrak{N}_2 | i \in \mathbb{N} \}$ of mutually orthogonal projections and families $\{ V_{1,i} \in \mathfrak{N}_1 | i \in \mathbb{N} \}$ and $\{ V_{2,i} \in \mathfrak{N}_2 | i \in \mathbb{N} \}$ of partial isometries such that $\sum_{i=0}^{\infty} E_{1,i} = \sum_{i=0}^{\infty} E_{2,i} = I$, $E_{1,i}=V_{1,i}^{*}V_{1,i}$, $E_{1,i+1}=V_{1,i}V_{1,i}^{*}$, $E_{2,i}=V_{2,i}^{*}V_{2,i}$ and $E_{2,i+1}=V_{2,i}V_{2,i}^{*}$ for any $i \in \mathbb{N}$, where $I$ is an identity operator on $\mathcal{H}$ \cite[Proposition 2.2.4]{Sakai1971}. 

By an assumption, $E_{1,n}E_{2,n} \neq 0$ for any $n \in \mathbb{N}$. Thus there is a unit vector $\Phi_n$ such that $E_{1,n}E_{2,n}\Phi_n=\Phi_n$. Let $\Psi$ be a unit vector in a Hilbert space $\mathcal{H}$ and
\[ \Psi_n := \left( 1 - \frac{1}{n} \right)^{1/2} \frac{(I - E_{1,n} -E_{1,n+1})(I-E_{2,n}-E_{2,n+1})}{\| (I - E_{1,n} -E_{1,n+1})(I-E_{2,n}-E_{2,n+1}) \|} \Psi + \left( \frac{1}{2n} \right)^{1/2} (\Phi_n+V_{1,n}V_{2,n}\Phi_n). \]
Then $\Psi_n$ is a unit vector and $\lim_{n \rightarrow \infty} \Psi_n = \Psi$ since $\sum_{i=0}^{n-1} E_{1,i} \leq I-E_{1,n}-E_{1,n+1}$, $\sum_{i=0}^{n-1} E_{2,i} \leq I -E_{2,n}-E_{2,n+1}$, and $\sum_{i=0}^{\infty} E_{1,i} = \sum_{i=0}^{\infty} E_{2,i} =I$. Let
\[ F_{1,n}:=\frac{1}{2}(E_{1,n}+E_{1,n+1}+V_{1,n}+V_{1,n}^*), \]
\[ F_{2,n}:=\frac{1}{2}(E_{2,n}+E_{2,n+1}+V_{2,n}+V_{2,n}^*). \]

Then $F_{1,n}$ and $F_{2,n}$ are projections in $\mathfrak{N}_1$ and $\mathfrak{N}_2$ respectively, and 
\[ E_{1,n}F_{1,n}\Psi_n=\frac{1}{2}\left( \frac{1}{2n} \right)^{1/2}(\Phi_n+V_{2,n}\Phi_n),\]
\[ F_{1,n}E_{1,n}\Psi_n=\frac{1}{2}\left( \frac{1}{2n} \right)^{1/2}(\Phi_n + V_{1,n}\Phi_n). \] 
Since $\| V_{1,n}\Phi_n - V_{2,n}\Phi_n \|^2=\langle V_{1,n} \Phi_n, V_{1,n} \Phi_n \rangle + \langle V_{2,n} \Phi_n, V_{2,n} \Phi_n \rangle - \langle V_{1,n} \Phi_n, V_{2,n} \Phi_n \rangle - \langle V_{2,n} \Phi_n, V_{1,n} \Phi_n \rangle = 2$, $V_{1,n}\Phi_n \neq V_{2,n}\Phi_n$. Thus $[E_{1,n},F_{1,n}]\Psi_n \neq 0$. Similarly $[E_{2,n},F_{2,n}]\Psi_n \neq 0$. Because
\[ E_{1,n}\Psi_n=E_{2,n}\Psi_n=\left( \frac{1}{2n} \right)^{1/2}  \Phi_n, \]
\[ F_{1,n}\Psi_n=F_{2,n}\Psi_n=\frac{1}{2}\left( \frac{1}{2n} \right)^{1/2} (\Phi_n+V_{1,n}\Phi_n+V_{2,n}\Phi_n+V_{1,n}V_{2,n}\Phi_n), \]
the vector state of $\mathfrak{N}_1 \vee \mathfrak{N}_2$ induced by $\Psi_n$ is an EPR state for incommensurable pairs $(E_{1,n},E_{2,n})$ and $(F_{1,n},F_{2,n})$.

\end{proof}


Theorem \ref{violates} and Theorem \ref{exists} entail the following corollary. This is a theorem of Halvorson and Clifton \cite[Proposition 1]{HalvorsonClifton2000}. 

\begin{corollary}[Halvorson-Clifton]
\label{HC}
Let $\mathfrak{N}_1$ and $\mathfrak{N}_2$ be properly infinite von Neumann algebras on a Hilbert space $\mathcal{H}$ such that $\mathfrak{N}_1 \subseteq \mathfrak{N}_2'$ and $X_1X_2 \neq 0$ for any nonzero operators $X_1 \in \mathfrak{N}_1$ and $X_2 \in \mathfrak{N}_2$.

For any unit vector $\Psi$ and any real number $\epsilon > 0$, there exists a unit vector $\Psi'$ such that $\| \Psi - \Psi' \| < \epsilon$ and a vector state induced by $\Psi'$ is Bell correlated.
\end{corollary}

\section{Concluding remarks}
In this paper, we have defined EPR states for incommensurable pairs which allow for similar arguments to EPR's, and clarified the logical relations between EPR states for incommensurable pairs and Bell correlated states. Let $\omega$ be a normal state of $\mathfrak{N}_1 \vee \mathfrak{N}_2$. Then the following relations hold (Corollary \ref{maximal_cor} and Theorem \ref{violates}).
\begin{center}
$\omega$ is strongly maximally Bell correlated.

$\Downarrow$ $\not \Uparrow$

$\omega$ is an EPR state for incommensurable pairs.

$\Downarrow$ $\not \Uparrow$

$\omega$ is Bell correlated.
\end{center}

In Theorem \ref{exists} it is shown that there is a norm dense set of EPR states for incommensurable pairs between two strictly space-like separated regions. This theorem entails that there is a norm dense set of Bell correlated states between two strictly space-like separated regions (Corollary \ref{HC}).

\section*{Acknowledgements}
The author is supported by the JSPS KAKENHI, No.23701009 and the John Templeton Foundation Grant ID 35771.



\begin{thebibliography}{99}
%
%
\bibitem{ArensVaradarajan2000} Arens, R. and Varadarajan, V. S.: On the concept of Einstein-Podolsky-Rosen states and their structure. J. Math. Phys. 41,638--651 (2000).
\bibitem{Baumgartel1995}
Baumg\"artel, H.: Operatoralgebraic Methods in Quantum Field Theory. Akademie Verlag, Berlin (1995)
\bibitem{BohataHamhalter2009} Bohata, M. and Hamhalter, J.: Maximal violation of Bell's inequalities and Pauli spin matrices. J. Math. Phys. 50, 082101 (2009).
\bibitem{BohataHamhalter2010} Bohata, M. and Hamhalter, J.: Bell's correlations and spin systems. Found. Phys. 40, 1065--1075.
\bibitem{Cirel'son1980} Cirel'son, B. S.: Quantum generalizations of Bell's inequality. Lett. Math. Phys. 4, 93--100 (1980).
\bibitem{EPR} Einstein, A.,
Podolsky, B., and Rosen, N.: Can quantum-mechanical description of physical reality
be considered complete? Phys. Rev. 47, 777--780 (1935).
\bibitem{Gisin1991} Gisin, N.: Bell's inequality holds for all non-product states. Phys. Lett. 154, 201--202.
\bibitem{HalvorsonClifton1999} Halvorson, H. and Clifton, R.: Maximal beable subalgebras of quantum mechanical observables. Int. J. Theor. Phys. 38, 2441--2484 (1999).
\bibitem{Halvorson2000} Halvorson, H.: The Einstein-Podolsky-Rosen state maximally violates Bell's inequalities. Lett. Math. Phys. 53, 321--329 (2000)
\bibitem{HalvorsonClifton2000} Halvorson, H. and Clifton, R.: Generic Bell correlation between arbitrary local algebras in quantum field theory. J. Math. Phys. 41, 1711--1717 (2000).
\bibitem{Hofer-SzaboVecsernyes2012a} Hofer-Szab\'{o}, G., Vecserny\'{e}s, P.: Noncommuting local common causes for correlations violating the Clauser-Horne inequality. J. Math. Phys. 53, 122301 (2012).
\bibitem{Hofer-SzaboVecsernyes2012b} Hofer-Szab\'{o}, G., Vecserny\'{e}s, P.: Bell inequality and common causal explanation in algebraic quantum field theory. arXiv:1204.5708 [quant-ph] (2012).
\bibitem{Huang2008} Huang, S.: On states of perfect correlation. J. Math. Phys. 49, 112101 (2008).
\bibitem{Ozawa2006} Ozawa, M.: Quantum perfect correlations. Ann. Phys. 321, 744--769 (2006).
\bibitem{OzawaKitajima2012} Ozawa, M. and Kitajima, Y.: Reconstructing Bohr's reply to EPR in algebraic quantum theory. Found. Phys. 42, 475--487 (2012).
\bibitem{Landau1987}
Landau, L. J.: On the violation of Bell's inequality in quantum theory. Phys. Lett. 120, 54--56 (1987)
\bibitem{Redei1998} R\'{e}dei, M.: Quantum logic in algebraic approach. Kluwer Academic, Dordrecht/Boston/London (1998).
\bibitem{Sakai1971} Sakai, S.: C*-algebra and W*-algebras. Springer, New York (1971).
\bibitem{Summers1990} Summers, S. J.: On the independence of local algebras in quantum field theory. Rev. Math. Phys. 2, 201-247 (1990).
\bibitem{SummersWerner1987a} Summers, S. J. and Werner, R.: Bell's inequalities and quantum field theory. I. General setting. J. Math. Phys. 28, 2440--2447 (1987).
\bibitem{Werner1999} Werner, R. F.: EPR states for von Neumann algebras. arXiv:quant-ph/9910077 (1999). 
\end{thebibliography}
\end{document}